\documentclass[english,12pt]{iopart}
\usepackage[T1]{fontenc}
\usepackage[latin9]{inputenc}
\usepackage{units}
\usepackage{amsthm}

\makeatletter
\usepackage{iopams}
\usepackage{setstack}
\theoremstyle{plain}
\newtheorem{thm}{Theorem}
\theoremstyle{definition}
\newtheorem{defn}[thm]{Definition}
\theoremstyle{plain}
\newtheorem{prop}[thm]{Proposition}
 \ifx\proof\undefined\
   \newenvironment{proof}[1][\proofname]{\par
     \normalfont\topsep6\p@\@plus6\p@\relax
     \trivlist
     \itemindent\parindent
     \item[\hskip\labelsep
           \scshape
       #1]\ignorespaces
   }{%
     \endtrivlist\@endpefalse
   }
   \providecommand{\proofname}{Proof}
 \fi
\theoremstyle{remark}
\newtheorem{rem}[thm]{Remark}
\theoremstyle{plain}
\newtheorem{lem}[thm]{Lemma}
\theoremstyle{plain}
\newtheorem{cor}[thm]{Corollary}

\makeatother

\usepackage{babel}
\eqnobysec
\begin{document}

\title[Eigenvalue repulsion estimates...]{Eigenvalue repulsion estimates and some applications for the one-dimensional
Anderson model}

\author{Alexander Rivkind}

\address{Physics Department, Technion - Israel Institute of Technology, Haifa
32000, Israel.}

\ead{sashkarivkind@gmail.com}

\author{Yevgeny Krivolapov}

\address{Physics Department, Technion - Israel Institute of Technology, Haifa
32000, Israel.}

\ead{evgkr@tx.technion.ac.il}

\author{Shmuel Fishman}

\address{Physics Department, Technion - Israel Institute of Technology, Haifa
32000, Israel.}

\ead{fishman@physics.technion.ac.il}

\author{Avy Soffer}

\address{Mathematics Department, Rutgers University, New-Brunswick, NJ 08903,
USA.}

\ead{soffer@math.rutgers.edu}
\begin{abstract}
We show that the spacing between eigenvalues of the discrete 1D Hamiltonian
with \textit{arbitrary potentials} which are bounded, and with Dirichlet
or Neumann Boundary Conditions is bounded away from zero. We prove
an explicit lower bound, given by $Ce^{-bN}$, where $N$ is the lattice
size, and $C$ and $b$ are some finite constants. In particular,
the spectra of such Hamiltonians have no degenerate eigenvalues. As
applications we show that to leading order in the coupling, the solution
of a nonlinearly perturbed Anderson model in one-dimension (on the
lattice) remains exponentially localized, in probability and average
sense for initial conditions given by a unique eigenfunction of the
linear problem. We also bound the derivative of the eigenfunctions
of the linear Anderson model with respect to a potential change.
\end{abstract}
\maketitle

\section{Introduction}

We consider the one dimensional Anderson model on the lattice, $\Lambda$,\begin{equation}
H_{\omega}^{\Lambda}u_{n}\left(x\right)=u_{n}\left(x+1\right)+u_{n}\left(x-1\right)+\varepsilon_{x}u_{n}\left(x\right)=E_{n}u_{n}\left(x\right),\end{equation}
 with $x,n\in\mathbb{Z}$, $\omega=\left\{ \varepsilon_{x}\right\} $
is the realization of the potential, $H_{\omega}^{\Lambda}$ is the
Hamiltonian on the domain $\Lambda$, with eigenfunctions $\left\{ u_{n}\left(x\right)\right\} \in L^{2}\left(\Lambda\right)$
and eigenvalues $E_{n}$. We will also denote by $N\equiv\left|\Lambda\right|$
the size of the domain. Furthermore, $H_{\omega}^{\Lambda}$ satisfies
some boundary conditions to be specified later, which include both
the Dirichlet and Neumann boundary conditions. Our first result applies
to \emph{arbitrary} uniformly bounded potential,\begin{equation}
\sup_{x\in\Lambda}\left|\varepsilon_{x}\right|\equiv W<\infty.\end{equation}
 We will show in the next section that the minimal distance between
the eigenvalues of $H_{\omega}^{\Lambda}$ is bounded below by a constant
of order $e^{-bN}$ for every $\omega$ and as long as $W<\infty$
and the boundary conditions defining $H_{\omega}^{\Lambda}$ are of
the allowed class. Note, that this result holds for \emph{all} bounded
potentials. Our proof, while not necessarily the simplest one, is
instructive and may be of more general interest.

Then, in the next section we show two applications, motivated by the
study of Anderson localization problem, both linear and nonlinear.
In particular, in \cite{FKS1,FKS2} we have shown that for the nonlinearly
perturbed Anderson model, \begin{equation}
i\partial_{t}\psi=H_{\omega}^{\Lambda}\psi+\beta\left|\psi\right|^{2}\psi,\label{eq:NLS}\end{equation}
 with the initial condition of $\psi\left(x,0\right)=u_{0}\left(x\right)$,
the first order nonlinear correction to the solution is given by,\begin{equation}
\psi^{\left(1\right)}\left(x,t\right)=\beta\sum_{n}c_{n}^{\left(1\right)}\left(t\right)u_{n}\left(x\right)e^{-iE_{n}t},\label{eq:intr_psi1}\end{equation}
 with\begin{equation}
c_{n}^{\left(1\right)}\left(t\right)=\frac{V_{n}^{000}}{E_{n}-E_{0}}\left(1-e^{i\left(E_{n}-E_{0}\right)t}\right).\label{eq:intr_c1}\end{equation}
 Higher order corrections involve products of $c_{n}^{\left(1\right)}$
and other combinations of energies. Relevant estimates were recently
proven for such combinations in \cite{AW} . Note, that since $H_{\omega}^{\Lambda}$
depends on the realization of the potential, $\omega$, so is $u_{0}\left(x\right)$.
We will show here that on average, the fractional power of the solution
of \eref{eq:NLS} to the first order in $\beta,$ remains exponentially
bounded for all times, we also show that the ordinary average is exponentially
bounded at least for times which are exponential in $N$.

In the second application we control the averages of fractional powers
of the derivative of the eigenfunctions of $H_{\omega}^{\Lambda}$
with respect to some $\varepsilon_{x}$, and show that they are exponentially
small in the distance between $x$ and the localization center of
the eigenfunction. We also bound the averages by some power of the
volume of the system and interpolate between the fractional and ordinary
averages.

The proof of the eigenvalue repulsion is based on the transfer matrix
representation of the solutions of the one dimensional problem, and
study the dependence of the eigenfunctions on the energy \cite{ishi}.
By studying the properties of the matrices as transformations of the
Hyperbolic space, in terms of the complex energy as a parameter, the
presence and absence of continuous spectrum for classes of Random
Schr\"odinger operators on graphs can be naturally analyzed \cite{FHS}.
This is close to our approach. We then show, that the condition for
the energy parameter to have a value that corresponds to an eigenvalue,
requires a path to return to the starting point, in some sense. Next,
we prove monotonicity of a rotation number/angle associated with the
path, as a function of the energy parameter. Monotonicity with respect
to the energy parameter, is also used in the hyperbolic space representation;
there, it appears as a basic property of the Mobius transformation
\cite{BHS}. By bounding the rate of rotation from above, as a function
of the energy parameter, a minimal distance between the eigenvalues
follows.

The applications mentioned above, use the exponentially small minimal
distance between the eigenvalues in a crucial way. We decompose dyadically
the space of potentials, $\omega$, to subsets where the minimal distance
between eigenvalues is in a dyadic interval, $I_{m}\in\left[2^{-m-1},2^{-m}\right]$.
Then, the sum over $m$ is bounded up to, $m\leq\bar{b}N$, due to
the eigenvalue repulsion. We estimate each term by a combination of
two probabilistic estimates: first, the Minami estimate for the probability
to find at least two eigenvalues in an interval $I$, \cite{Mi}\begin{equation}
\Pr\left(TrP_{H_{\omega}^{\Lambda}}^{(\Lambda)}\left(I\right)\geq2\right)\leq\left(\pi\left\Vert \rho\right\Vert _{\infty}I\, N\right)^{2},\label{eq:Minami0}\end{equation}
 where $\left\Vert \rho\right\Vert _{\infty}$ is the supremum of
the density of states, $I$ is some energy interval while $P_{H_{\omega}^{\Lambda}}^{(\Lambda)}(I)$
is the spectral projection on that interval and $H_{\omega}^{\Lambda}$
is the Hamiltonian corresponding to a one-dimensional Anderson problem
with Dirichlet boundary conditions on a domain $\Lambda$. The second
bound we use is the fractional moment bound of Aizenman \cite{A}
(see also related bounds in \cite{AG,GK}), \begin{equation}
\left\langle \sum_{n}\left|u_{n}\left(x\right)u_{n}\left(y\right)\right|\right\rangle \leq De^{-\mu\left|x-y\right|},\label{eq:Aizenman}\end{equation}
 where $\mu>0$, and $D>0$ are some constants.

\section{Lower bound on level spacings}

\subsection{\label{sec:Main-Result}Main Result}

The Main Result is the following Theorem:
\begin{thm}
(eigenvalue repulsion)\label{main_thm} Given the tight binding model:
\[
Hu_{n}=u_{n-1}+\varepsilon_{n}u_{n}+u_{n+1};\qquad1\le n\le N\]
 with Dirichlet boundary conditions ($u_{0}=0$, $u_{N+1}=0$) or
Neumann boundary conditions ($u_{0}=u_{1}$, $u_{N+1}=u_{N}$) and
with $0\le\varepsilon_{n}\le W<\infty$ for all n, there exists a
constant 0$<$$\eta(W)$$<$1 such that: $\left|E_{i}-E_{j}\right|\ge\frac{\pi(\eta^{-1}(W)-1)}{\eta^{-N}(W)-1}$
are for all $i\ne j$ and eigenvalues $E_{i},E_{j}$ (in this Section
as well as in \ref{sec:AppA} and \ref{sec:AppB}
lattice sites are denoted by $n$).
\end{thm}

\subsection{Setup}

For simplicity, we first prove the main Theorem for Dirichlet b.c.,
and then describe the modifications needed for Neumann case in a separate
subsection. Obviously,\textit{ E} is an eigenvalue of \textit{H} ,
if and only if, there exists a non-trivial vector $\vec{u}=\left\{ u_{n}\right\} _{n=0}^{N+1}$
so that $H\vec{u}=E\vec{u}$ and: \begin{eqnarray}
 & u_{0} & =0\label{eq_1a}\\
 & u_{N+1} & =0.\label{eq_1b}\end{eqnarray}
 Since $u_{0}=u_{1}=0$ implies $\vec{u}\equiv0$, we can set (without
loss of generality): \begin{equation}
u_{1}=1.\label{eq_1c}\end{equation}
 For arbitrary \textit{E}, given \eref{eq_1a}, \eref{eq_1c} we
can calculate all the components $u_{n}$ of $\vec{u}$ by recursive
formula: \begin{equation}
u_{n+1}(E)=(E-\varepsilon_{n})u_{n}(E)-u_{n-1}(E).\label{eq_2}\end{equation}
 $E$ is an eigenvalue of $H$ iff \eref{eq_1b} holds.
\begin{defn}
\label{phi_def_gen} Let $\alpha_{n}(E)$ be the angle between the
2D vector $(u_{n-1}(E),u_{n}(E)),$ and the positive direction of
abscissa at the corresponding Cartesian plane. $\varphi_{n}(E)$ is
said to be a \emph{version of angular ratio} between $u_{n-1}(E)$
and $u_{n}(E),$ iff $\exists m\in\mathbb{Z},$ so that $\alpha_{n}(E)+2m\pi=\varphi_{n}(E)$.
\end{defn}
By the definition \ref{phi_def_gen}, it holds for $k\in\mathbb{Z}$
: \begin{equation}
\varphi_{N+1}=k\pi\Leftrightarrow u_{N+1}=0\Leftrightarrow E-eigenvalue.\label{eq_3}\end{equation}
 Also:
\begin{enumerate}
\item If $u_{n-1}\ne0$ then: $\tan(\varphi_{n}):=\frac{u_{n}}{u_{n-1}}$,
\item If $u_{n}\ne0$ then: $\cot(\varphi_{n}):=\frac{u_{n-1}}{u_{n}}$.
\end{enumerate}
By the recursive formula \eref{eq_2}, we have:
\begin{equation}
\fl
\varphi_{n+1}(E)=
\left\{
\begin{array}{ll}
\arctan\left(E-\varepsilon_{n}-\cot\varphi_{n}(E)\right)+(k+2m)\pi  & k\pi<\varphi_{n}(E)<(k+1)\pi\\
\left(k-\frac{1}{2}\right)\pi+2m\pi                                 &  \varphi_{n}(E)=k\pi
\end{array}
\right.
\label{eq_4}
\end{equation}

where $k,m\in\mathbb{Z}$ and $\arctan:\,\mathbb{R}\to\left(-\frac{\pi}{2},\frac{\pi}{2}\right)$.

The upper row of \eref{eq_4} can be obtained by dividing \eref{eq_2}
by $u_{n}$, identifying $\cot(\varphi_{n}):=\frac{u_{n-1}}{u_{n}}$,
taking \textit{arctan} of both sides, and then adding arbitrary integer
number of full rotations\textit{ m}. (In what follows $m=0$ will
be chosen so that $\varphi_{n}$ are differentiable as functions of
$E$.)

The lower row of \eref{eq_4}, refers to the case when $\varphi_{n}(E)=k\pi\Leftrightarrow u_{n}=0\Rightarrow u_{n+1}=-u_{n-1}$.
Then, it is obvious by considering the directions of 2D vectors $(u_{n-1},0)$
and $(0,-u_{n+1})$. Once again, arbitrary integer number of full
rotations is added.

By \eref{eq_4}, the sequence $\left\{ \varphi_{n}\right\} $ is
well defined up to addition of $2m\pi$, $m\in\mathbb{N}$. Therefore,
any version of \eref{eq_4} can be used, to verify if the condition
\eref{eq_3} holds for an energy \textit{E.} In particular, we can
set $m\equiv0$ and choose $\left\{ \varphi_{n}\right\} _{n=1}^{N+1}$
to be:

\begin{equation}
\fl
\begin{array}{l}
\varphi_{1}(E)\equiv\frac{\pi}{2}\\
\varphi_{n+1}(E)=\left\{ \begin{array}{ll}
\arctan\left(E-\varepsilon_{n}-\cot\varphi_{n}(E)\right)+k\pi & k\pi<\varphi_{n}(E)<(k+1)\pi\\
\varphi_{n+1}(E)=\left(k-\frac{1}{2}\right)\pi & \varphi_{n}(E)=k\pi\end{array}\right.
\end{array}\label{eq_5}
\end{equation}

The version \eref{eq_5} of \eref{eq_4} is especially convenient
for our further use, because $\varphi_{N+1}(E)$ turns to be a continuously
differentiable function of $E$ (see Proposition \ref{prop1}).

The angle variable $\varphi_{n}$ is known as the Pr\"uffer angle \cite{KLS}%
\footnote{We thank Michael Aizenman for bringing this to our attention%
}.

\subsection{Proof}

\subsubsection{Proof for Dirichlet boundary conditions}

Eigenvalues satisfy \eref{eq_3} as explained in the setup\textbf{
}section. We will show that $\varphi_{N+1}(E)$ rotates monotonously
counterclockwise (Proposition\textbf{ \ref{prop1}}), and that there
is no degeneracy (Proposition\textbf{ \ref{prop2})}.

We then show that the rotation speed $\varphi'_{N+1},$ is bounded
from above (Proposition\textbf{ \ref{prop3}}). But $\varphi_{N+1}(E)$
must change by angle of $\pi$ between every pair of eigen-energies,
(see \eref{eq_3}), and its rotation speed is bounded from above,
therefore, the spacing between eigenvalues is bounded from below;
that is: $|E_{i+1}-E_{i}|\ge\frac{\pi}{\varphi_{N+1}^{'\max}}$.
\begin{prop}
\label{prop1} $\varphi_{N+1}(E)$ is a continuously differentiable
and a strictly increasing function of E. \end{prop}
\begin{proof}
$\varphi_{2}(E)=\arctan(E-\varepsilon_{1})$ is continuously differentiable
and $\varphi'_{2}(E)>0$ . Next we use induction in \textit{n}:

If $\varphi_{n}\ne k\pi,$ then $\varphi_{n+1}(E)$ is continuously
differentiable and increasing, according to the definition, since
\textit{arctan(.)} is a strictly increasing differentiable function
of its argument. The argument of \textit{arctan(.)}: $E-\cot(\varphi_{n})$
is strictly increasing and continuously differentiable (by induction
assumption on $\varphi_{n}$ starting from $\varphi_{2}$).

If $\varphi_{n}=k\pi,$ and continuously differentiable and increasing
(with respect to $E$), then: $\varphi_{n+1}(E)$ is continuous, because
the single side limits \eref{lim_phi_pi_minus} and \eref{lim_phi_pi_plus}
are equal to each other as follows from the definition \eref{eq_5}:
\begin{eqnarray}
\lim_{\varphi_{n}\to k\pi^{-}}\varphi_{n+1}(E)=\pi(k-1)+\lim_{\varphi_{n}\to k\pi^{-}}\arctan(-\cot(\varphi_{n}))=\nonumber \\
=\pi(k-1)+\arctan(-(-\infty))=k\pi-\frac{\pi}{2}\label{lim_phi_pi_minus}\\
\lim_{\varphi_{n}\to k\pi^{+}}\varphi_{n+1}(E)=k\pi+\lim_{\varphi_{n}\to k\pi^{+}}\arctan(-\cot(\varphi_{n}))=\nonumber \\
=k\pi+\arctan(-(+\infty))=k\pi-\frac{\pi}{2}\label{lim_phi_pi_plus}\end{eqnarray}
 left and right {}``single-sided'' derivatives of $\varphi_{n+1}(E)$
exist and equal. That is because, for any $\varphi_{n}\ne k\pi$ point,
it holds: \begin{equation}
\varphi'_{n+1}(E)=\frac{d\varphi_{n+1}(E)}{dE}=\frac{1+\frac{\varphi'_{n}}{\sin^{2}\varphi_{n}}}{1+(E-\varepsilon_{n}-\cot\varphi_{n})^{2}}\label{eq_6}\end{equation}
 Taking the single side limits at $\varphi_{n}=k\pi,$ and recalling
that $\varphi_{n+1}$ is continuous, one gets: \begin{equation}
\begin{array}{l}
\lim_{\varphi_{n}\to k\pi^{+}}\varphi'_{n+1}(E)=\lim_{\varphi_{n}\to k\pi^{-}}\varphi'_{n+1}(E)^{}=\varphi'_{n}(E)\\
\Rightarrow\varphi'_{n+1}(E)=\varphi'_{n}(E).\end{array}\label{eq_7}\end{equation}
 Hence derivative exists, and (by induction) is positive as required. \end{proof}
\begin{prop}
\label{prop2}The spectrum of H is simple.
\end{prop}
General proof of the simplicity of spectrum is given in theorem \ref{simplicity_thm}.
Simplicity of spectrum can also be shown using {}``$\varphi$''
(Pruffer angle \cite{KLS}) formalism used here:
\begin{proof}
To ensure that no degeneracy occurs (that is $\left|E_{i}-E_{j}\right|\ne0$),
we need to show that solutions of $\varphi_{N+1}(E)=k\pi$ are simple.
It is sufficient to show that $\varphi_{N+1}'(E)\ne0$.

We saw that $\varphi'_{2}(E)>0$ consequently:

For $\varphi_{n}(E)\ne0,$ we have that: $\varphi'_{n+1}(E)=\frac{1+\frac{\varphi'_{n}}{\sin^{2}\varphi_{n}}}{1+(E-\varepsilon_{n}-\cot\varphi_{n})^{2}}>0$
(by induction assumption). \label{eq_8}

For $\varphi_{n}(E)=0$, $\varphi'_{n+1}(E)=\varphi'_{n}(E)>0$ (by
induction assumption).

Therefore $\varphi_{N+1}'(E)>0$ and there is no degeneracy.

Another way to ensure absence of degenerate eigenvalues of \textit{H}
is by successive use of \eref{eq_5}, and considering limits of $\varphi_{N+1}(E)$
at $E\to\pm\infty$: \begin{equation}
\begin{array}{ll}
\lim_{E\to-\infty}\varphi_{2}(E)=-\frac{\pi}{2} & \Rightarrow\lim_{E\to-\infty}\varphi_{3}(E)=-\frac{3\pi}{2}....\Rightarrow\\
 & \Rightarrow\lim_{E\to-\infty}\varphi_{N+1}(E)=-\left(2N-1\right)\frac{\pi}{2},\end{array}\end{equation}
 \begin{equation}
\lim_{E\to+\infty}\varphi_{2}(E)=\frac{\pi}{2}\Rightarrow\lim_{E\to+\infty}\varphi_{3}(E)=\frac{\pi}{2}....\Rightarrow\lim_{E\to+\infty}\varphi_{N+1}(E)=\frac{\pi}{2}\end{equation}
 By continuity and monotonicity of $\varphi_{N+1}(E),$ there exist
exactly \textit{N }different solutions of $\varphi_{N+1}(E)=k\pi$
for $E\in(-\infty,+\infty)$. \end{proof}
\begin{prop}
\label{prop3} The ratio of derivatives $\frac{\varphi'_{n+1}(E)}{\varphi'_{n}(E)}$
is bounded above.\end{prop}
\begin{proof}
By \eref{eq_6} : \begin{equation}
\begin{array}{l}
\varphi'_{n+1}(E)=\frac{d\varphi_{n+1}(E)}{dE}=\frac{1}{1+(E-\varepsilon_{n}-\cot\varphi_{n})^{2}}+\frac{\frac{\varphi'_{n}}{\sin^{2}\varphi_{n}}}{1+(E-\varepsilon_{n}-\cot\varphi_{n})^{2}}\le\\
\le1+\frac{\frac{\varphi'_{n}}{\sin^{2}\varphi_{n}}}{1+(E-\varepsilon_{n}-\cot\varphi_{n})^{2}}=1+\frac{\varphi'_{n}}{1-2(E-\varepsilon_{n})\sin\varphi_{n}\cos\varphi_{n}+(E-\varepsilon_{n})^{2}\sin^{2}\varphi_{n}}\end{array}\label{eq_9}\end{equation}
 It is left to find a lower bound for the denominator:\begin{equation}
\begin{array}{l}
q=1-2(E-\varepsilon_{n})\sin\varphi_{n}\cos\varphi_{n}+(E-\varepsilon_{n})^{2}\sin^{2}\varphi_{n}=\\
=1-(E-\varepsilon_{n})\sin2\varphi_{n}+(E-\varepsilon_{n})^{2}\sin^{2}\varphi_{n}.\end{array}\label{eq_10}\end{equation}
 For convenience, we define $x:=E-\varepsilon_{n},$ then: \begin{equation}
q(x,\varphi_{n})=1-2x\sin\varphi_{n}\cos\varphi_{n}+x^{2}\sin^{2}\varphi_{n}.\label{eq_qx}\end{equation}
 We are only interested in cases where $\left|x\right|=\left|E-\varepsilon_{n}\right|\le W+1$,
because otherwise \textit{E} is out of the spectrum interval, and
cannot be in an interval between any pair of eigenstates. Furthermore:
\eref{eq_10} is $\pi$-periodic in $\varphi_{n}$. Therefore, we
are looking at a bound on the \textit{q}, on the \textit{compact}
set (closed rectangle): \begin{equation}
A=\left\{ \left|x\right|\le W+1\right\} \times\left\{ 0\le\varphi_{n}\le\pi\right\} .\label{eq_defA}\end{equation}

Due to continuity of \textit{q} in both \textit{x} and $\varphi_{n}$,
it has minimum in $A$. Therefore, to show that $q$ is bounded away
from zero, it is sufficient to prove that $q$ is positive in $A$.
(the formal continuity of \textit{q} in arguments ($x,\varphi_{n}$)
is obvious, and not to be confused with continuity of functions $\varphi_{n}$,
$\varphi'_{n}$\textit{ }with respect to \textit{x }or \textit{E}).

Given the expression \eref{eq_qx} with fixed $\varphi_{n}$, it
can be evaluated as quadratic function in \textit{x }with minimum
value: \begin{equation}
\min_{x}\left\{ 1-2x\sin\varphi_{n}\cos\varphi_{n}+x^{2}\sin^{2}\varphi_{n}\right\} =\sin^{2}\varphi_{n}.\label{eq_17}\end{equation}
 Therefore $q(x,\varphi_{n})$ is positive for any $\varphi_{n}\ne k\pi,$
but also positive for $\varphi_{n}=k\pi,$ since $\varphi_{n}=k\pi\Rightarrow q=1$.
We hence have: \begin{equation}
\eta(W):=\min_{\left|E-\varepsilon_{n}\right|\le W+1,\varphi_{n}}\left\{ q\right\} >0.\label{eq_etadef}\end{equation}
 (Calculation of an analytic expression for $\eta$ is given in appendix
\ref{A_eta}. It demonstrates the dependence of the bound on $W$.)
Substituting into the derivative inequality \eref{eq_9}, we obtain
a recursive inequality:\begin{equation}
\varphi'_{n+1}(E)\le1+\frac{\varphi'_{n}}{\eta}.\label{eq_15}\end{equation}
 \end{proof}
\begin{rem}
No separate argument is required for $\varphi_{n}=k\pi$, since limit
\eref{eq_7} exists, and is a particular case of \eref{eq_9}.\end{rem}
\begin{proof}
{[} Proof of theorem \ref{main_thm} for Dirichlet b.c.{]}

Evaluating the inequality \eref{eq_15} recursively, from $\varphi'_{2}(E)\le1$
to $\varphi_{N+1}(E)$ \textit{(recall $\varphi_{2}(E)=\arctan(E-\varepsilon_{1})$)}
one obtains:

\begin{equation}
\begin{array}{l}
\varphi'_{N+1}(E)\le\underbrace{(1+\frac{1}{\eta}(1+\frac{1}{\eta}(1+\frac{1}{\eta}(...(1+\frac{1}{\eta}(1+\frac{1}{\eta}(1+\frac{\varphi'_{2}}{\eta}))...)}_{N-1}\\
\le\underbrace{(1+\frac{1}{\eta}(1+\frac{1}{\eta}(1+\frac{1}{\eta}(...(1+\frac{1}{\eta}(1+\frac{1}{\eta}(1+\frac{1}{\eta})))...)}_{N-1}=\frac{\eta^{-N}-1}{\eta^{-1}-1}\end{array}\label{eq_16}\end{equation}
 Since no degeneracy occurs (By proposition \ref{prop2}) and since
eigenvalues satisfy \eref{eq_3}, and the derivative $\varphi'_{N+1}(E)$
satisfies \eref{eq_16} for any $-1\le E\le W+1$, we have:

$\left|E_{i}-E_{j}\right|\ge\frac{\pi(\eta^{-1}(W)-1)}{\eta^{-N}(W)-1}$
for any pair of eigenstates $E_{i},E_{j}$.
\end{proof}

\subsubsection{Proof for Neumann boundary conditions}

Neumann b.c. imply: \begin{equation}
u_{1}=u_{0}\ne0\Rightarrow\varphi_{1}(E)\equiv\frac{\pi}{4}.\label{eq_phiNeu}\end{equation}
 The \eref{eq_3} is modified to be: \begin{equation}
\varphi_{N+1}=(k+1/4)\pi\Leftrightarrow u_{N+1}=u_{N}\Leftrightarrow E-eigenvalue.\label{eq_phiNeu}\end{equation}
 The derivative $\varphi_{2}(E)'>0,$ and therefore proposition \ref{prop1}
holds. Limits in proposition \ref{prop2} hold starting from $\varphi_{2}(E)$.
Proposition \ref{prop2} holds without modifications. Thus theorem
\ref{main_thm} holds for Neumann b.c. as well.

\subsubsection{Simplicity of the spectrum for general boundary conditions}
\begin{thm}
\label{simplicity_thm} Let $H$ be defined as in theorem \ref{main_thm}
with boundary conditions such that the normalization determines the
value of the eigenfunction at two adjacent points. Then the spectrum
of $H$ is simple. \end{thm}
\begin{proof}
Let $\vec{v}$ and $\vec{u}$ be eigenvectors of $H$ with eigenvalue
$E_{0}$. Without loss of generality $\vec{v}$ and $\vec{u}$, can
be normalized so that $u_{1}=v_{1}=a$ and $u_{0}=v_{0}=b$. But form
this point all the remaining elements of $\vec{v}$ and $\vec{u}$
can be determined using \eref{eq_2}. Therefore $\vec{v}\equiv\vec{u}$.
\end{proof}

\subsection{Remarks on level repulsion}
\begin{enumerate}
\item As expected, the limit on spacing vanishes when \textit{W }approaches
infinity. That is, when nearly infinite barriers are allowed (See
\eref{eq_13}, \eref{eq_14}).
\item The typical sensitivity to the energy of the angle variable $\varphi_{N+1}$
is exponentially large in $N$ at the proximity of an eigenvalue (See
appendix \ref{A_phi})
\item The result does not hold for periodic boundary conditions. For periodic
boundary conditions, degeneracy might occur, and there is \textit{no}
lower bound on the energy spacing between non-degenerate states. For
example consider $\varepsilon_{j}\equiv0$. The eigenvalues \[
\begin{array}{cc}
E_{j}=2\cos\frac{2\pi j}{N}\qquad & j=1..N\end{array}\]
 are pair wise degenerate, except \textit{$j=N$ }and $j=N/2$\textit{
}(in case of even \textit{$N$}).
\item This work does not prove that the proposed limit is optimal. However,
when considering exponential localization in disordered potentials,
it appears that optimal bound on inter-level spacings is indeed exponential
in chain length $N$.
\end{enumerate}

\section{Applications}

\subsection{Bound on first order term in perturbation theory}

In this Section we will show some applications of Theorem \ref{main_thm}.
Our main interest will be in problems related to the still open question
of whether there is localization for the Anderson model perturbed
by a small nonlinearity. The numerical results so far are inconclusive,
see e.g. \cite{Molina1998,Kopidakis2008,Pikovsky2008,Flach2009,Skokos2009}.
The only rigorous result that applies to the full nonlinear system,
is the finite time result of \cite{WZ}. In recent works \cite{FKS1,FKS2,FKS3},
we developed a renormalized perturbation expansion for the nonlinear
problem. The first order term, can now be controlled rigorously, as
we will show now. The correction is given by the following term :
(see \eref{eq:intr_psi1} and \eref{eq:intr_c1}),

\begin{equation}
c_{n}^{\left(1\right)}=V_{n}^{000}\left(\frac{1-e^{i\left(E_{n}-E_{0}\right)t}}{E_{n}-E_{0}}\right)\label{eq:c1}\end{equation}
 where,\begin{equation}
V_{n}^{000}=\sum_{y}u_{n}\left(y\right)u_{0}^{3}\left(y\right),\label{eq:Vn000}\end{equation}
 We define,\begin{equation}
\psi^{\left(1\right)}\left(x,t\right)=\sum_{n}c_{n}^{\left(1\right)}\left(t\right)u_{n}\left(x\right).\end{equation}
This is the correction to the wave function in first order perturbation
theory \cite{FKS1} (in this section lattice sites are denoted by
$x$ and $y$, while eigenstates by $n$). Following \cite{AM} we
are interested in bounding the fractional moments of $\left|\psi^{\left(1\right)}\right|$,
\begin{equation}
\left|\psi^{\left(1\right)}\left(x,t\right)\right|^{s}\leq\sum_{n}\left|c_{n}^{\left(1\right)}\left(t\right)u_{n}\left(x\right)\right|^{s},\end{equation}
 for $0<s<1$. Namely,\begin{eqnarray}
\left\langle \left|\psi^{\left(1\right)}\left(x,t\right)\right|^{s}\right\rangle  & \leq\sum_{n}\langle\left|c_{n}^{\left(1\right)}\left(t\right)u_{n}\left(x\right)\right|^{s}\rangle\nonumber\\
 & \leq\sum_{y,n}\int d\mu\left(\omega\right)\left|\frac{u_{n}\left(x\right)u_{n}\left(y\right)u_{0}^{3}\left(y\right)}{E_{n}-E_{0}}\right|^{s},\label{eq:psi1_initial_bound}\end{eqnarray}
 where $d\mu\left(\omega\right)$ is the measure which is defined
on the random potentials.
\begin{lem}
\label{lem:psi_aprio_bnd}For any Hamiltonian in a finite box $\Lambda$
of size $N,$ with orthonormal eigenfunctions, $\sum_{x}u_{n}\left(x\right)u_{m}\left(x\right)=\delta_{n,m}$.
For any $x,$ \begin{equation}
\exists n,\qquad\left|u_{n}\left(x\right)\right|\geq N^{-\nicefrac{1}{2}}.\end{equation}
 \end{lem}
\begin{proof}
Assume that $\exists x$ such that $\forall n$, $\left|u_{n}\left(x\right)\right|<N^{-\nicefrac{1}{2}}.$
Due to orthonormality of the eigenfunctions of a Hamiltonian,\[
1=\sum_{n}\left|u_{n}\left(x\right)\right|^{2}<N^{-1}\sum_{n}1=1,\]
 which is in a contradiction to the assumption.\end{proof}
\begin{thm}
\label{thm:psi1_bound}For the one-dimensional Anderson model, and
for $0<s<\frac{1}{5}$,\begin{equation}
\left\langle \left|\psi^{\left(1\right)}\left(x,t\right)\right|^{s}\right\rangle \leq C_{s}N^{\nicefrac{9}{4}-\nicefrac{s}{2}}e^{-\mu s\left|x\right|},\end{equation}
 where $C_{s}$ is a constant which depends only on $s$.\end{thm}
\begin{proof}
Lets define a set of potentials with the help of the dyadic decomposition,
\begin{equation}
\mathcal{V}_{n}\left(m\right)=\left\{ \omega|\left|E_{n}-E_{0}\right|\in I_{m}\right\} ,\label{eq:dyadic_decomp}\end{equation}
 where,\begin{equation}
I_{m}\equiv\left[2^{-m-1},2^{-m}\right].\end{equation}
 The denominators of \eref{eq:psi1_initial_bound} cannot be arbitrarily
small by Theorem \ref{main_thm} (eigenvalue repulsion),\begin{equation}
\left|E_{n}-E_{0}\right|\geq Ce^{-bN},\label{eq:eig_min_space}\end{equation}
 where $C$ and $b$ are constants. Therefore combining the decomposition
\eref{eq:dyadic_decomp} with \eref{eq:psi1_initial_bound} and
\eref{eq:eig_min_space} yields, \begin{equation}\fl
\left\langle \left|\psi^{\left(1\right)}\left(x,t\right)\right|^{s}\right\rangle \leq\sum_{m=0}^{M}2^{s\left(m+1\right)}\sum_{y,n}d\mu\left(\omega\right)\chi\left(\mathcal{V}_{n}\left(m\right)\right)\left|u_{n}\left(x\right)u_{n}\left(y\right)u_{0}^{3}\left(y\right)\right|^{s},\end{equation}
 with,\begin{equation}
\frac{1}{\ln2}\left(b N-\ln C\right)\leq \bar{b}N \equiv M,\label{eq:M_def}\end{equation}
 and $\chi\left(\mathcal{V}_{m}\left(n\right)\right)$ is the characteristic
function of the set of potentials, $\mathcal{V}_{m}\left(n\right)$.
Applying the generalized H\"older inequality one finds,\begin{eqnarray}
\left\langle \left|\psi^{\left(1\right)}\left(x,t\right)\right|^{s}\right\rangle  & \leq\sum_{n,y}\sum_{m=0}^{M}2^{s\left(m+1\right)}\left(\int d\mu\left(\omega\right)\chi\left(\mathcal{V}_{m}\left(n\right)\right)^{p_{1}}\right)^{\nicefrac{1}{p_{1}}},\label{eq:psi1_1}\nonumber\\
 & \times\left(\int d\mu\left(\omega\right)\left|u_{n}\left(x\right)u_{n}\left(y\right)\right|^{sp_{2}}\right)^{\nicefrac{1}{p_{2}}}\nonumber \\
 & \times\left(\int d\mu\left(\omega\right)\left|u_{0}\left(y\right)u_{0}\left(0\right)u_{0}^{-1}\left(0\right)\right|^{3sp_{3}}\right)^{\nicefrac{1}{p_{3}}}\end{eqnarray}
 with $\frac{1}{p_{1}}+\frac{1}{p_{2}}+\frac{1}{p_{3}}=1$. To estimate,
\begin{equation}
J_{1}=\left(\int d\mu\left(\omega\right)\chi\left(\mathcal{V}_{m}\left(n\right)\right)^{p_{1}}\right)^{\nicefrac{1}{p_{1}}},\end{equation}
 we use the fact that $\chi^{p_{1}}=\chi$ and Minami estimate for
the probability to find at least two eigenvalues in an interval $I$
(see \eref{eq:Minami0}) \cite{Mi}\begin{equation}
\Pr\left(TrP_{H_{\omega}^{\Lambda}}^{(\Lambda)}\left(I\right)\geq2\right)\leq\left(\pi\left\Vert \rho\right\Vert _{\infty}I\, N\right)^{2}.\label{eq:Minami}\end{equation}
 Since, we are not interested in a particular energy interval, we
will cover the energy band with $I^{-1}$ intervals of size $I=\left|I_{m}\right|=2^{-m-1}$,
which gives,\begin{equation}
J_{1}\leq\left(\pi\left\Vert \rho\right\Vert _{\infty}N\right)^{\nicefrac{2}{p_{1}}}2^{-\nicefrac{\left(m+1\right)}{p_{1}}}.\label{eq:J1_bound}\end{equation}
 To bound,\begin{equation}
J_{2}=\left(\int d\mu\left(\omega\right)\left|u_{n}\left(x\right)u_{n}\left(y\right)\right|^{sp_{2}}\right)^{\nicefrac{1}{p_{2}}},\end{equation}
 and \begin{equation}
J_{3}=\left(\int d\mu\left(\omega\right)\left|u_{0}\left(y\right)u_{0}\left(0\right)u_{0}^{-1}\left(0\right)\right|^{3sp_{3}}\right)^{\nicefrac{1}{p_{3}}},\end{equation}
 we choose, $1/p_{2}=s$ and $1/p_{3}=3s$, which sets, $1/p_{1}=1-4s$
and $s<\frac{1}{4}$. Then, we proceed by combining Lemma \ref{lem:psi_aprio_bnd}
to bound, $\left|u_{0}\left(0\right)\right|^{-1}$, with the result
of Aizenman \cite{A},\begin{equation}
\sum_{n}\left(\int d\mu\left(\omega\right)\left|u_{n}\left(x\right)u_{n}\left(y\right)\right|\right)\leq De^{-\mu\left|x-y\right|},\label{eq:Aizenman1}\end{equation}
 where $\mu>0$, and $D>0$ are some constants. This yields,\begin{equation}
J_{2}\leq D^{s}e^{-\mu s\left|x-y\right|},\label{eq:J2_bound}\end{equation}
 and\begin{equation}
J_{3}\leq N^{\nicefrac{3s}{2}}D^{3s}e^{-3\mu s\left|y\right|}.\label{eq:J3_bound}\end{equation}
 Plugging \eref{eq:J1_bound}, \eref{eq:J2_bound} and \eref{eq:J3_bound}
back into \eref{eq:psi1_1} gives,\begin{equation}
\fl\left\langle \left|\psi^{\left(1\right)}\left(x,t\right)\right|^{s}\right\rangle \leq D^{4s}\left(\pi\left\Vert \rho\right\Vert _{\infty}\right)^{2\left(1-4s\right)}N^{3-\nicefrac{13s}{2}}\left(\sum_{m=0}^{M}2^{\left(5s-1\right)\left(m+1\right)}\right)\sum_{y}e^{-\mu s\left|x-y\right|}e^{-3\mu s\left|y\right|}.\label{eq:psi1_2}\end{equation}
 Setting $s<\frac{1}{5}$, and using the triangle inequality for the
last sum, we get, \begin{equation}
\sum_{m=0}^{M}2^{\left(5s-1\right)\left(m+1\right)}<\frac{1}{1-2^{\left(5s-1\right)}},\label{eq:sum_bound}\end{equation}
 and,\begin{equation}
\sum_{y}e^{-\mu s\left|x-y\right|}e^{-3\mu s\left|y\right|}\leq e^{-\mu s\left|x\right|}\sum_{y}e^{-2\mu s\left|y\right|}\leq e^{-\mu s\left|x\right|}\frac{1}{1-e^{-2\mu s}}.\end{equation}
 Therefore for $0<s<\frac{1}{5}$,\begin{eqnarray}
\left\langle \left|\psi^{\left(1\right)}\left(x,t\right)\right|^{s}\right\rangle  & \leq C_{s,\delta} & N^{3-\nicefrac{13s}{2}}e^{-\mu s\left|x\right|},\label{eq:psi1_almost_final_bound}\end{eqnarray}
 with,\[
C_{s,\delta}=\frac{D^{4s}\left(\pi\left\Vert \rho\right\Vert _{\infty}\right)^{2\left(1-4s\right)}}{\left(1-2^{\left(5s-1\right)}\right)\left(1-e^{-2\mu s}\right)}.\]
 \end{proof}
\begin{cor}
For $\nu>0$ and $0<s<\frac{1}{5}$,\begin{equation}
\Pr\left(\left|\psi^{\left(1\right)}\left(x,t\right)\right|\geq C_{s}^{\nicefrac{1}{s}}N^{\nicefrac{3}{s}-\nicefrac{13}{2}}e^{-\left(\mu-\nu/s\right)\left|x\right|}\right)\leq e^{-\nu\left|x\right|}.\end{equation}
 \end{cor}
\begin{proof}
Using Chebychev inequality,\begin{equation}
\Pr\left(\left|\psi^{\left(1\right)}\left(x,t\right)\right|^{s}\geq A\right)\leq A^{-1}C_{s}N^{3-\nicefrac{13s}{2}}e^{-\mu s\left|x\right|},\end{equation}
 and choosing,\begin{equation}
A=C_{s}N^{3-\nicefrac{13s}{2}}e^{-\left(\mu-\nu\right)s\left|x\right|},\end{equation}
 gives,\begin{equation}
\Pr\left(\left|\psi^{\left(1\right)}\left(x,t\right)\right|^{s}\geq C_{s}N^{3-\nicefrac{13s}{2}}e^{-\left(\mu s-\nu\right)\left|x\right|}\right)\leq e^{-\nu\left|x\right|},\end{equation}
 or\begin{equation}
\Pr\left(\left|\psi^{\left(1\right)}\left(x,t\right)\right|\geq C_{s}^{\nicefrac{1}{s}}N^{\nicefrac{3}{s}-\nicefrac{13}{2}}e^{-\left(\mu-\nu/s\right)\left|x\right|}\right)\leq e^{-\nu\left|x\right|}.\end{equation}

\end{proof}

\subsection{Bound on the derivative of an eigenfunction}

An important object in the study of the properties of eigenfunctions,
is the sensitivity to a change of the potential at some point of an
eigenfunction. We have by direct computation,\begin{equation}
\frac{\partial u_{0}\left(x\right)}{\partial\varepsilon_{y}}=u_{0}\left(y\right)\sum_{n\ne0}\frac{u_{n}\left(x\right)u_{n}\left(y\right)}{E_{0}-E_{n}}.\end{equation}
 The above analysis could be extended to obtain bounds on this derivative.
\begin{thm}
\label{thm:derivative_s_1_3}For a one-dimensional Anderson problem,
and $0<s<\frac{1}{3}$,\begin{equation}
E_{s}\equiv\left\langle \left|\frac{\partial u_{0}\left(x\right)}{\partial\varepsilon_{y}}\right|^{s}\right\rangle \leq K_{s}N^{3-\nicefrac{7s}{2}}e^{-\mu s\left|x-y\right|}e^{-\mu s\left|y\right|}.\end{equation}
 \end{thm}
\begin{proof}
Proceeding in a similar manner to the previous subsection,\begin{eqnarray}
 E_{s}&\leq\sum_{n\ne0}\left\langle \frac{\left|u_{0}\left(y\right)u_{n}\left(x\right)u_{n}\left(y\right)\right|^{s}}{\left|E_{0}-E_{n}\right|^{s}}\right\rangle\nonumber \\ &\leq\sum_{n\ne0}\sum_{m=0}^{M}2^{s\left(m+1\right)}\int d\mu\left(\omega\right)\chi\left(\mathcal{V}_{n}\left(m\right)\right)\left|u_{0}\left(y\right)u_{n}\left(x\right)u_{n}\left(y\right)\right|^{s}.\end{eqnarray}
 Now use the generalized H\"older inequalty,\begin{eqnarray}
 E_{s} & \leq\sum_{n\ne0}\sum_{m=0}^{M}2^{s\left(m+1\right)}\left(\int d\mu\left(\omega\right)\chi^{p_{1}}\left(\mathcal{V}_{n}\left(m\right)\right)\right)^{\nicefrac{1}{p_{1}}}\left(\int d\mu\left(\omega\right)\left|u_{n}\left(x\right)u_{n}\left(y\right)\right|^{sp_{2}}\right)^{\nicefrac{1}{p_{2}}}\nonumber\\
 & \times\left(\int d\mu\left(\omega\right)\left|u_{0}\left(y\right)u_{0}\left(0\right)u_{0}\left(0\right)^{-1}\right|^{sp_{3}}\right)^{\nicefrac{1}{p_{3}}},\label{eq:E1_00} \end{eqnarray}
 with $\frac{1}{p_{1}}+\frac{1}{p_{2}}+\frac{1}{p_{3}}=1$. Setting
$1/p_{2}=1/p_{3}=s$, we have $1/p_{1}=1-2s$. Than using \eref{eq:J1_bound},\begin{eqnarray}
E_{s} & \leq\left(\pi\left\Vert \rho\right\Vert _{\infty}N\right)^{2\left(1-2s\right)}\left(\sum_{m=0}^{M}2^{\left(3s-1\right)\left(m+1\right)}\right)\label{eq:sum_bound2}\\
 & \times\sum_{n\ne0}\left(\int d\mu\left(\omega\right)\left|u_{n}\left(x\right)u_{n}\left(y\right)\right|\right)^{s}\left(\int d\mu\left(\omega\right)\left|u_{0}\left(y\right)u_{0}\left(0\right)u_{0}\left(0\right)^{-1}\right|\right)^{s}.\nonumber \end{eqnarray}
 Bounding the sum for $s<\frac{1}{3}$, utilizing the result of Aizenman
\eref{eq:Aizenman1} and Lemma \ref{lem:psi_aprio_bnd} for the last
term gives, \[
E_{s}\leq K_{s}N^{3-\nicefrac{7s}{2}}e^{-\mu s\left|x-y\right|}e^{-\mu s\left|y\right|},\]
 with\[
K_{s}=\frac{D^{2s}\left(\pi\left\Vert \rho\right\Vert _{\infty}\right)^{2\left(1-2s\right)}}{1-2^{1-3s}}.\]

\end{proof}

\subsection{Ordinary averages, $s=1$}

Note, that in the previous bounds the result of Theorem \ref{main_thm}
were not used, since $s$ was selected such that the sums \ref{eq:sum_bound}
and \ref{eq:sum_bound2} were convergent even for unbounded $M$.
In the following we will calculate the bounds on the ordinary averages,
namely for $s=1$, and then interpolate between those two results.
\begin{thm}
\label{thm:psi1_s_1}For \textup{the one-dimensional Anderson model
on a box $\Lambda$ of size $N$, and $\psi^{\left(1\right)}\left(x,t\right)$
defined in \eref{eq:intr_psi1},}\begin{equation}
\left\langle \left|\psi^{\left(1\right)}\left(x,t\right)\right|\right\rangle \leq A\, N^{\nicefrac{13}{2}}.\end{equation}
 \end{thm}
\begin{proof}
For the first order correction of the wavefunction we get (substituting
$s=1$ in \eref{eq:psi1_1}),\begin{eqnarray}
\left\langle \left|\psi^{\left(1\right)}\left(x,t\right)\right|\right\rangle  & \leq\sum_{n,y}\sum_{m=0}^{M}2^{\left(m+1\right)}\left(\int d\mu\left(\omega\right)\chi\left(\mathcal{V}_{m}\left(n\right)\right)^{p_{1}}\right)^{\nicefrac{1}{p_{1}}},\\
 & \times\left(\int d\mu\left(\omega\right)\left|u_{n}\left(x\right)u_{n}\left(y\right)\right|^{p_{2}}\right)^{\nicefrac{1}{p_{2}}}\nonumber \\
 & \times\left(\int d\mu\left(\omega\right)\left|u_{0}\left(y\right)u_{0}\left(0\right)u_{0}^{-1}\left(0\right)\right|^{3p_{3}}\right)^{\nicefrac{1}{p_{3}}}\nonumber \end{eqnarray}
 Using the bounds \eref{eq:J1_bound}, \eref{eq:J2_bound} and \eref{eq:J3_bound}
gives, \begin{eqnarray}
\left\langle \left|\psi^{\left(1\right)}\left(x,t\right)\right|\right\rangle & \leq D^{1/p_{2}+1/p_{3}}\left(\pi\left\Vert \rho\right\Vert _{\infty}\right)^{\nicefrac{2}{p_{1}}}N^{\nicefrac{5}{2}+\nicefrac{2}{p_{1}}}\nonumber \\
& \left(\sum_{m=0}^{M}2^{\left(1-\nicefrac{1}{p_{1}}\right)\left(m+1\right)}\right)\sum_{y}e^{-\mu\left|x-y\right|/p_{2}}e^{-\mu\left|y\right|/p_{3}}.\label{eq:psi1_exp0}\end{eqnarray}
 Setting, $1/p_{1}=1-\epsilon$ and $1/p_{2}=1/p_{3}=\epsilon/2$
and using the triangle inequality for the last sum, yields\begin{equation}
\left\langle \left|\psi^{\left(1\right)}\left(x,t\right)\right|\right\rangle \leq\frac{D^{\epsilon}\left(\pi\left\Vert \rho\right\Vert _{\infty}\right)^{2\left(1-\epsilon\right)}}{\left(1-e^{-\mu\epsilon/2}\right)\left(1-2^{-\epsilon}\right)}N^{\nicefrac{5}{2}+2\left(1-\epsilon\right)}2^{M\epsilon}.\label{eq:psi1_exp}\end{equation}
 Since $M=\bar{b}N$ \eref{eq:M_def}we will set $\epsilon=1/N$
to remove the exponential dependence on $N$, this gives,\begin{equation}
\left\langle \left|\psi^{\left(1\right)}\left(x,t\right)\right|\right\rangle \leq\frac{D^{1/N}2^{\bar{b}}\left(\pi\left\Vert \rho\right\Vert _{\infty}\right)^{2\left(1-1/N\right)}}{\left(1-e^{-\mu/\left(2N\right)}\right)\left(1-2^{-1/N}\right)}N^{\nicefrac{5}{2}+2\left(1-1/N\right)},\end{equation}
 or,\begin{equation}
\left\langle \left|\psi^{\left(1\right)}\left(x,t\right)\right|\right\rangle \leq A\, N^{\nicefrac{13}{2}},\end{equation}
 with\begin{equation}
A=\frac{2^{\bar{b}+1}\pi^{2}\left\Vert \rho\right\Vert _{\infty}^{2}D}{\mu\ln2}.\end{equation}
 \end{proof}
\begin{thm}
\textup{\label{thm:derivative_s_1}For the one-dimensional Anderson
model on a box $\Lambda$ of size $N$,}\begin{equation}
\left\langle \left|\frac{\partial u_{0}\left(x\right)}{\partial\varepsilon_{y}}\right|\right\rangle \leq BN^{\nicefrac{11}{2}},\label{eq:thm_13}\end{equation}
 \end{thm}
\begin{proof}
Similarly, for the bound on the derivative, found from \eref{eq:E1_00}
by substituting $s=1$, \begin{eqnarray}
 E_{1} & \leq\sum_{n\ne0}\sum_{m=0}^{M}2^{\left(m+1\right)}\left(\int d\mu\left(\omega\right)\chi^{p_{1}}\left(\mathcal{V}_{n}\left(m\right)\right)\right)^{\nicefrac{1}{p_{1}}}\left(\int d\mu\left(\omega\right)\left|u_{n}\left(x\right)u_{n}\left(y\right)\right|^{p_{2}}\right)^{\nicefrac{1}{p_{2}}}\nonumber\\
 & \times\left(\int d\mu\left(\omega\right)\left|u_{0}\left(y\right)u_{0}\left(0\right)u_{0}\left(0\right)^{-1}\right|^{p_{3}}\right)^{\nicefrac{1}{p_{3}}}, \end{eqnarray}
 we get,\begin{equation}
\fl E_{1}\leq D^{\nicefrac{1}{p_{2}}+\nicefrac{1}{p_{3}}}\left(\pi\left\Vert \rho\right\Vert _{\infty}\right)^{\nicefrac{2}{p_{1}}}N^{\nicefrac{5}{2}+\nicefrac{2}{p_{1}}}\left(\sum_{m=0}^{M}2^{\left(1-\nicefrac{1}{p_{1}}\right)\left(m+1\right)}\right)e^{-\mu\left|x-y\right|/p_{2}}e^{-\mu\left|y\right|/p_{3}},\end{equation}
 setting as before $1/p_{1}=1-\epsilon$, $1/p_{2}=1/p_{3}=\epsilon/2$
gives,\begin{equation}
E_{1}\leq\frac{D^{\epsilon}\left(\pi\left\Vert \rho\right\Vert _{\infty}\right)^{2\left(1-\epsilon\right)}}{\left(1-2^{-\epsilon}\right)}N^{\nicefrac{9}{2}-2\epsilon}2^{\epsilon M}e^{-\mu\epsilon\left|x-y\right|/2}e^{-\mu\epsilon\left|y\right|/2}.\end{equation}
 Since $M=\bar{b}N$ \eref{eq:M_def} we will set $\epsilon=1/N$
to remove the exponential dependence on $N$,\begin{equation}
E_{1}\leq\frac{D^{1/N}2^{\bar{b}}\left(\pi\left\Vert \rho\right\Vert _{\infty}\right)^{2\left(1-1/N\right)}}{\left(1-2^{-1/N}\right)}N^{\nicefrac{9}{2}-\nicefrac{2}{N}}e^{-\mu\left|x-y\right|/\left(2N\right)}e^{-\mu\left|y\right|/\left(2N\right)},\end{equation}
 or\begin{equation}
\left\langle \left|\frac{\partial u_{0}\left(x\right)}{\partial\varepsilon_{y}}\right|\right\rangle \leq BN^{\nicefrac{11}{2}},\end{equation}
 with \begin{equation}
B=D2^{\bar{b}}\left(\pi\left\Vert \rho\right\Vert _{\infty}\right)^{2}.\end{equation}
 \end{proof}
\begin{rem}
The result above is not optimal, in fact it can be shown, using a
different argument, that \eref{eq:thm_13} can be improved to be
linear in $N$. This will be shown elsewhere.\end{rem}
\begin{cor}
\textup{For the one-dimensional Anderson model on a box $\Lambda$
of size $N$, and $0<s\leq1$,}\begin{equation}
\left\langle \left|\psi^{\left(1\right)}\left(x,t\right)\right|^{s}\right\rangle \leq A_{s}N^{\left(11s+3\right)/2}e^{-2\mu\left(1-s\right)\left|x\right|/9}.\end{equation}
 and\begin{equation}
\left\langle \left|\frac{\partial u_{0}\left(x\right)}{\partial\varepsilon_{y}}\right|^{s}\right\rangle \leq B_{s}N^{5s+1}e^{-2\mu\left(1-s\right)\left|x-y\right|/5}e^{-2\mu\left(1-s\right)\left|y\right|/5}.\end{equation}
 \end{cor}
\begin{proof}
Since $f\left(s\right)\equiv\left\langle \left|.\right|^{s}\right\rangle $
is a holomorfic and bounded function for $0<s\leq1$, we utilize Hadamard
three-line interpolation theorem. For the first order correction to
the wavefunction using Theorem \ref{thm:psi1_bound} for $s=2/11$
and Theorem \ref{thm:psi1_s_1}, gives, \begin{equation}
\left\langle \left|\psi^{\left(1\right)}\left(x,t\right)\right|^{s}\right\rangle \leq\left(C_{2/11}N^{\nicefrac{95}{44}}e^{-2\mu\left|x\right|/11}\right)^{\vartheta}\left(A\, N^{\nicefrac{13}{2}}\right)^{1-\vartheta},\end{equation}
 with,\begin{equation}
\vartheta=\frac{11}{9}\left(1-s\right).\end{equation}
 Leading to,\begin{equation}
\left\langle \left|\psi^{\left(1\right)}\left(x,t\right)\right|^{s}\right\rangle \leq A_{s}N^{\nicefrac{\left(191s+43\right)}{36}}e^{-2\mu\left(1-s\right)\left|x\right|/9}\leq A_{s}N^{\left(11s+3\right)/2}e^{-2\mu\left(1-s\right)\left|x\right|/9}.\end{equation}
 Similarly, for the derivative of the eigenfunction combining Theorem
\ref{thm:derivative_s_1_3} for $s=2/7$ and Theorem \ref{thm:derivative_s_1},
gives\[
\left\langle \left|\frac{\partial u_{0}\left(x\right)}{\partial\varepsilon_{y}}\right|^{s}\right\rangle \leq\left(K_{2/7}N^{2}e^{-2\mu\left|x-y\right|/7}e^{-2\mu\left|y\right|/7}\right)^{\vartheta}\left(BN^{\nicefrac{11}{2}}\right)^{1-\vartheta},\]
 with\begin{equation}
\vartheta=\frac{7}{5}\left(1-s\right).\end{equation}
 Or,\begin{eqnarray}
 \left\langle \left|\frac{\partial u_{0}\left(x\right)}{\partial\varepsilon_{y}}\right|^{s}\right\rangle & \leq B_{s}N^{\nicefrac{\left(49s+6\right)}{10}}e^{-2\mu\left(1-s\right)\left|x-y\right|/5}e^{-2\mu\left(1-s\right)\left|y\right|/5}\nonumber\\
 & \leq B_{s}N^{5s+1}e^{-2\mu\left(1-s\right)\left|x-y\right|/5}e^{-2\mu\left(1-s\right)\left|y\right|/5}.\end{eqnarray}
 \end{proof}
\begin{rem}
This Collorary sugests exponential bounds for both the first order
correction of the wavefunction and the derivative of the eigenfunction
for the whole range $0<s<1$.
\end{rem}

\subsection{Time dependent bound}

In this subsection we will eliminate the exponential dependence on
the volume (see \eref{eq:psi1_exp}) of the bound on the average
first order correction to the wavefunction, $\psi^{\left(1\right)}$,
by using the apriori bound\begin{equation}
\left|\frac{1-e^{i\left(E_{n}-E_{0}\right)t}}{E_{n}-E_{0}}\right|\leq t.\end{equation}

\begin{thm}
\textup{For the one-dimensional Anderson model on a box $\Lambda$
of size $N$, and $\epsilon_{0}\left(t\right)<\epsilon<1$, such that
$\epsilon_{0}\left(t\right)\rightarrow0$ for $t\rightarrow\infty$,
and $t\leq2^{\bar{b}N}$ with $\bar{b}$ given by \eref{eq:M_def}.}\begin{equation}
\left\langle \left|\psi^{\left(1\right)}\left(x,t\right)\right|\right\rangle \leq K_{\epsilon}N^{\nicefrac{5}{2}+2\left(1-\epsilon\right)/3}t^{\left(2+\epsilon\right)/3}\log_{2}t\, e^{-\mu\left|x\right|/2}.\end{equation}
 \end{thm}
\begin{proof}
We start with\begin{eqnarray}
\left\langle \left|\psi^{\left(1\right)}\left(x,t\right)\right|\right\rangle & \leq D^{1/p_{2}+1/p_{3}}\left(\pi\left\Vert \rho\right\Vert _{\infty}\right)^{\nicefrac{2}{p_{1}}}N^{\nicefrac{5}{2}+\nicefrac{2}{p_{1}}}\nonumber\\
& \times \left(\sum_{m=0}^{M}2^{\left(1-\nicefrac{1}{p_{1}}\right)\left(m+1\right)}\right)\sum_{y}e^{-\mu\left|x-y\right|/p_{2}}e^{-\mu\left|y\right|/p_{3}}.\label{eq:psi1_exp01}\end{eqnarray}
 and set, $1/p_{2}=1/2-2\epsilon,$ $1/p_{3}=\left(1/2-\epsilon\right)$
and $1/p_{1}=3\epsilon$, for $0<\epsilon<1/3$. If one is not intrested
in the $t$ dependence of the bound one calculates,\begin{equation}
\fl \left\langle \left|\psi^{\left(1\right)}\left(x,t\right)\right|\right\rangle \leq2D^{1-3\epsilon}\frac{\left(\pi\left\Vert \rho\right\Vert _{\infty}\right)^{6\epsilon}}{1-e^{-\epsilon}}N^{\nicefrac{5}{2}+6\epsilon}\left(\sum_{m=0}^{M}2^{\left(m+1\right)\left(1-3\epsilon\right)}\right)e^{-\mu\left|x\right|\left(1/2-2\epsilon\right)}.\end{equation}
 To obtain the time-dependent bound we split the sum to two parts,\begin{equation}
S\equiv\sum_{m=0}^{M}2^{\left(1-3\epsilon\right)\left(m+1\right)}\leq\sum_{m=0}^{M_{1}}2^{\left(1-3\epsilon\right)\left(m+1\right)}+t\sum_{m=M_{1}+1}^{M}2^{-3\epsilon\left(m+1\right)},\end{equation}
 where we have used the fact that,\begin{equation}
\left|\frac{1-e^{i\left(E_{n}-E_{0}\right)t}}{E_{n}-E_{0}}\right|\leq\min\left(t,2^{m}\right),\end{equation}
 which defines, $M_{1}=\log_{2}t$. Therefore, for sufficiently large
$t$\begin{equation}
\fl S\leq2t^{\left(1-3\epsilon\right)}\log_{2}t+t\frac{2^{-3\epsilon\left(M_{1}+1\right)}}{1-2^{-3\epsilon}}\leq2t^{\left(1-3\epsilon\right)}\left(\log_{2}t+\frac{1}{1-2^{-3\epsilon}}\right)\leq3t^{\left(1-3\epsilon\right)}\log_{2}t,\end{equation}
 and\[
\left\langle \left|\psi^{\left(1\right)}\left(x,t\right)\right|\right\rangle \leq K_{\epsilon}N^{\nicefrac{5}{2}+6\epsilon}t^{\left(1-3\epsilon\right)}\left(\log_{2}t\right)e^{-\mu\left|x\right|\left(1/2-2\epsilon\right)},\]
 with \[
K_{\epsilon}=6D^{1-3\epsilon}\frac{\left(\pi\left\Vert \rho\right\Vert _{\infty}\right)^{6\epsilon}}{1-e^{-\epsilon}},\]
 and $t\leq2^{\bar{b}N}$.
\end{proof}
\ack
This work was partly supported by the Israel Science Foundation (ISF),
by the US-Israel Binational Science Foundation (BSF), by the USA National
Science Foundation (NSF DMS-0903651), by the Minerva Center of Nonlinear
Physics of Complex Systems, by the Shlomo Kaplansky academic chair
and by the Fund for promotion of research at the Technion.

\appendix

\section{\label{sec:AppA}Relation Between $\varphi_{N+1}$ and the Normalized
Amplitude}

\label{A_phi}
\begin{defn}
\label{A_phi_exp} we define $z_{n}(E):=\frac{u_{n}}{u_{n-1}}$.\end{defn}
\begin{cor}
$z_{n}(E)=\tan(\varphi_{n}(E))$ \end{cor}
\begin{lem}
\label{u_z_lemma} Let $E_{0}$ be an eigenvalue of $H$ with corresponding
normalized eigenvector $\{v_{n}\}_{n=1}^{N},$ then, $v_{N}^{2}=(\frac{dz_{N+1}(E)}{dE}|_{E=E_{0}})^{-1}$.\end{lem}
\begin{proof}
By differentiability of $\varphi_{N+1}$, $\frac{dz_{N+1}(E)}{dE}|_{E=E_{0}}$
exists. By \eref{eq_2} and \eref{eq_1b} it holds: \begin{equation}
0=z_{N+1}(E_{0})=E_{0}-\varepsilon_{N}-\frac{1}{z_{N}(E_{0})}.\label{eq_ZN1}\end{equation}
 (For Neumann b.c $z_{N+1}(E_{0})=1$, the rest of the argument still
applies.)\\
 Differentiating the LHS of \eref{eq_ZN1} w.r.t. $E$ and $\varepsilon_{N}$,
one obtains: \begin{equation}
\frac{\partial z_{N+1}}{\partial E}dE+  \frac{\partial z_{N+1}}{\partial\varepsilon_{N}}d\varepsilon_{N}=0.
\label{eq_diffZN1}\end{equation}
 Differentiating the RHS of \eref{eq_ZN1} one finds: $\frac{\partial z_{N+1}}{\partial\varepsilon_{N}}=-1$,
therefore using $\frac{dE}{d\varepsilon_{N}}|_{E=E_{0}}=v_{N}^{2}$
(resulting from Feynman-Hellman theorem) one finds: \begin{equation}
\frac{dz_{N+1}(E)}{dE}|_{E=E_{0}}=\frac{d\varepsilon_{N}}{dE}|_{E=E_{0}}=v_{N}^{-2}.\end{equation}
 \end{proof}
\begin{cor}
For a random potential with Anderson localization the derivative $\varphi_{n}',$
typically takes exponentially large (in $N$) values.
\end{cor}

\section{\label{sec:AppB}Expression for $\eta(W)$}

\label{A_eta}

The minimum of $q$ (as in \eref{eq_qx}) can be either at local
extremum point or at the boundaries of $A$ (defined by \eref{eq_defA});
we will check both cases:
\begin{enumerate}
\item local extrema in inner points: \begin{equation}
\fl
\left\{ \begin{array}{l}
\frac{\partial q}{\partial x}=0\Rightarrow\sin\varphi_{n}\cos\varphi_{n}=x\sin^{2}\varphi_{n}\quad\Longrightarrow\quad\cos\varphi_{n}=x\sin\varphi_{n}\qquad\left(\varphi_{n}\ne0,\pi\right)\\
\frac{\partial q}{\partial\varphi_{n}}=0\Rightarrow2x\cos2\varphi_{n}=x^{2}\sin2\varphi_{n}\quad\Longrightarrow\quad2\cos2\varphi_{n}=x\sin2\varphi_{n}\qquad\left(\varphi_{n}\ne0,\pi,\: x\ne0\right)\end{array}\right.\label{eq_q_grad}\end{equation}
 The case $x=0$ is not interesting: if indeed extrema obtained for
$x=0,$ then $q(x=0,\varphi_{n})\equiv1$. Elsewhere \eref{eq_q_grad}
implies $\sin^{2}\varphi=0$ $\left(\varphi=k\pi\right);$ consequently
it is not an inner point of $A$. That is: \textit{no} local extrema
in \textit{A }with $q(x,\varphi_{n})\ne1$.
\item Boundaries of $A$: \begin{eqnarray}
 & \varphi_{n}=0,\pi\qquad\Rightarrow\qquad q=1\\
 & x=W+1\qquad\Rightarrow\qquad\varphi_{n}^{extr}=\frac{1}{2}\arctan\frac{2}{\left(W+1\right)}\left(+\frac{\pi}{2}\right)\label{eq_12a}\\
 & x=-(W+1)\qquad\Rightarrow\qquad\varphi_{n}^{extr}=\frac{-1}{2}\arctan\frac{2}{\left(W+1\right)}\left(+\frac{\pi}{2}\right)\label{eq_12b}\end{eqnarray}
 (with superscript $^{extr}$ standing for an extremum value on the
appropriate boundary.)\\

\end{enumerate}
Substituting \eref{eq_12a} and \eref{eq_12b} in \eref{eq_qx}
gives the following four options:\begin{equation}
q^{extr}=\left\{ \begin{array}{l}
1-(W+1)\sin\left(\arctan\frac{2}{W+1}\right)+(W+1)^{2}\sin^{2}\left(\frac{1}{2}\arctan\frac{2}{W+1}\right)\\
1-(W+1)\sin\left(\arctan\frac{2}{W+1}\right)+(W+1)^{2}\cos^{2}\left(\frac{1}{2}\arctan\frac{2}{W+1}\right)\\
1-(W+1)\sin\left(\arctan\frac{2}{W+1}\right)+(W+1)^{2}\sin^{2}\left(\frac{1}{2}\arctan\frac{2}{W+1}\right)\\
1-(W+1)\sin\left(\arctan\frac{2}{W+1}\right)+(W+1)^{2}\cos^{2}\left(\frac{1}{2}\arctan\frac{2}{W+1}\right)\end{array}\right.\label{eq_qextr}\end{equation}
 By the arguments that lead to \eref{eq_etadef}, all the entries
of \eref{eq_qextr} are positive. The minimum of $q$ is hence:
\begin{eqnarray}
\eta(W) &=\min_{\left|E-\varepsilon_{n}\right|\le W+1,\varphi_{n}}\left\{ q\right\} \label{eq_13} \\
           &=1-(W+1)\sin\left(\arctan\frac{2}{W+1}\right)+ \nonumber\\
           &+(W+1)^{2}\min\left\{ \sin^{2}\left(\frac{1}{2}\arctan\frac{2}{W+1}\right),\cos^{2}\left(\frac{1}{2}\arctan\frac{2}{W+1}\right)\right\} \nonumber
\end{eqnarray}
Using \eref{eq_17} combined with \eref{eq_12a}, \eref{eq_12b}
one can obtain a simpler, yet weaker bound on $\eta$:

\begin{equation}
\eta(W)\ge\min\left\{ \sin^{2}\left(\frac{1}{2}\arctan\frac{2}{\left(W+1\right)}\right),\cos^{2}\left(\frac{1}{2}\arctan\frac{2}{\left(W+1\right)}\right)\right\} \label{eq_14}\end{equation}

\end{document}